\documentclass[11pt]{article}

\usepackage{color,graphicx}

\renewcommand{\tilde}[1]{\widetilde{#1}} %

\newcommand{\poly}{\mathrm{poly}}  
\newcommand{\negl}{\mathrm{negl}}

\def\01{\{0,1\}}

\newcommand{\calA}{\mathcal{A}}

\newcommand{\calO}{\mathcal{O}}

\newcommand{\eps}{\epsilon} %

\let\epsilon=\varepsilon %

\newcommand{\microspace}{\mspace{.5mu}} %
\newcommand{\ket}[1]{\ensuremath{\lvert\microspace #1
    \microspace\rangle}} %
\newcommand{\bra}[1]{\ensuremath{\langle\microspace #1
    \microspace\rvert}} %
\newcommand{\ketbra}[2]{\lvert #1 \rangle \! \langle #2 \rvert} %
\newcommand{\kb}[1]{\ketbra{#1}{#1}} %

\newcommand{\norm}[1]{\lVert#1\rVert}

\newcommand{\comment}[1]{}

\newcommand{\class}[1]{\textsf{\textup{#1}}\xspace} %

\newcommand{\NP}{\class{NP}} %
\newcommand{\QMA}{\class{QMA}} %
\newcommand{\QCMA}{\class{QCMA}} %
\newcommand{\UQMA}{\class{UQMA}} 
\newcommand{\BQP}{\class{BQP}} %
\newcommand{\MA}{\class{MA}} %
\newcommand{\QMAtwo}{\class{QMA(2)}}

\newcommand{\CQMAtwosepp}[1]{\class{cloneableQMA}^{\mathrm{SEP}}_{#1}(2)}
\newcommand{\CQMAtwop}[1]{\class{cloneableQMA}_{#1}(2)}  
\newcommand{\yes}{\mathrm{yes}}   
\newcommand{\no}{\mathrm{no}}   
\newcommand{\CQMA}{\class{cloneableQMA}}
\newcommand{\CQMAp}[1]{\class{cloneableQMA}_{#1}}
\newcommand{\CQMAtwo}{\class{cloneableQMA(2)}}  
\newcommand{\CQMAk}{\class{cloneableQMA}(k)} 
\newcommand{\CQMAkp}[1]{\class{cloneableQMA}_{#1}(k)} 

\usepackage{wrapfig}

\usepackage[margin=1in]{geometry} 
\usepackage[inline]{enumitem}%
\usepackage[utf8]{inputenc}
\usepackage{titlesec} 
\usepackage{amssymb} 
\usepackage{amsopn}
\usepackage{xspace}
\usepackage{graphicx} 
\usepackage{relsize} 
\usepackage{bm} 
\usepackage{xcolor} 
\usepackage{amsmath}
\usepackage{breqn} 
\usepackage{algpseudocode} 
\usepackage{multirow}  
\usepackage{amsthm} 
\usepackage{fullpage}
\usepackage{booktabs} %
\usepackage{dsfont} %
\usepackage{float} %
\usepackage{hyperref} %
\usepackage{mdframed} %
\usepackage{microtype} %
\usepackage{charter} %
\usepackage{soul} %
\usepackage{suffix} %
\usepackage{cleveref}
\usepackage{multicol} 
\usepackage{stackengine}
\usepackage{scalerel} 
\usepackage{titlesec}
\usepackage{tikz} 
  \usetikzlibrary{calc,positioning} 
\usepackage{mdframed}

\newtheorem{theorem}{Theorem}[section]
 
\newtheorem{lemma}[theorem]{Lemma} 

\newtheorem{corollary}[theorem]{Corollary}
\theoremstyle{definition}
\newtheorem{definition}[theorem]{Definition}
\newtheorem{remark}[theorem]{Remark}

\titlespacing*{\paragraph}{0pt}{4pt}{3pt} %

\title{
The Role of piracy in quantum proofs
}

\author{
Anne Broadbent\thanks{University of Ottawa, Canada \tt{abroadbe@uottawa.ca}}
\and 
Alex B. Grilo\thanks{Sorbonne Universit\'e, CNRS, LIP6, France. \tt{Alex.Bredariol-Grilo@lip6.fr}}
\and 
Supartha Podder\thanks{Stony Brook University, USA. \tt{supartha@cs.stonybrook.edu}}
\and 
Jamie Sikora\thanks{Virginia Tech, USA. \tt{sikora@vt.edu}} 
} 
   
\date{}

\begin{document} 

\maketitle

\begin{abstract}  
A well-known feature of quantum information is that it cannot, in general, be cloned. 
Recently, a number of quantum-enabled information-processing tasks have demonstrated  various forms of uncloneability; among these forms, \emph{piracy} is an adversarial model that gives maximal power to the adversary, in controlling both a \emph{cloning-type} attack, as well as the \emph{evaluation/verification} stage. 
Here, we initiate the study of \emph{anti-piracy} proof systems, which are proof systems that inherently prevent piracy attacks. We define anti-piracy proof systems, demonstrate such a proof system for an oracle problem, and 
also describe a candidate anti-piracy proof system for $\NP$. 
We also study quantum proof systems that are \emph{cloneable} and settle the famous \QMA vs.~$\QMA(2)$ debate in this setting. 
Lastly, we discuss how one can approach 
the \QMA vs.~\QCMA question, by studying its cloneable variants. 
\end{abstract}

\section{Introduction}
 The groundbreaking result of Goldwasser, Micali and Rackoff~\cite{GMR89} revolutionized the way we think about proofs, by   
 defining a proof system as an interactive game between an unbounded prover and a computationally bounded verifier, where the prover wants to convince the verifier about the truth of some statement. For a proof system to be valid, we want that there is a way that an honest prover can convince the verifier of a true statement, whereas we want that no malicious prover can convince the verifier of a false statement.

This foundational result is adopted as the perfect paradigm for exploring fundamental concepts in a many areas, such as in 
approximability~\cite{Arora98}, operator algebras~\cite{JNV+20arxiv}, and, highly relevant to this work, cryptography. Already in \cite{GMR89}, a cryptographic version of an interactive proof systems was defined as a \emph{zero-knowledge proof system}. This is an interactive proof system, with the additional property that the  verifier does not learn anything new except the fact that the proved statement is indeed true. Such a seemingly paradoxical object has had a profound impact to cryptography, and understanding which problems have zero-knowledge proof systems has had far-reaching applications of theoretical and practical interest. In particular, the definition of zero-knowledge triggered the development of the simulation paradigm, which is fundamental in modern cryptography~\cite{Gol04b,Lindel1l6}.

Starting with the definition of $\QMA$~\cite{KSV02} as the quantum analogue of~$\NP$, where a {\em quantum} polynomial-time algorithm verifies a {\em quantum} proof, the study of proof systems in the quantum world has led to foundational results towards the understanding of quantum information. The  notion of a quantum proof has been extended to the interactive setting~\cite{Wat03}, with multiple provers~\cite{KM03,CHTW04}, 
and zero-knowledge~\cite{watrous2002limits}.

Simultaneously to the above, the area of quantum cryptography has flourished in its study of how quantum resources help in cryptographic protocols. A key property that enables many of these protocols is that quantum states cannot be generically copied~\cite{WZ82,Die82}. We note, however that concrete quantum cryptography protocols involve specific distributions of quantum states and that proof techniques for these schemes go well beyond the basic no-cloning principle.  
Fortunately, many use cases for quantum-cloning in cryptography have been developed, together with their corresponding security notions and proofs. This includes 
quantum money~\cite{Wie83,MolinaVW12}  and new functionalities such as encryption and functionalities with  certified deletion and copy-protection (see \Cref{sec:related-work-anti-piracy}).

Given the importance of proof systems in several domains, including cryptography, and the importance of the uncloneability in quantum cryptography, our work seeks to unify the two theories, by studying uncloneability in proof systems. Our central question is of the existence of such proofs. 

\begin{center}
\label{eqn:CQ1}
\text{\em \textbf{Central Question:} Which problems have uncloneable quantum proof systems, and which do not?} 
\end{center}

This topic was raised over ten years ago by Aaronson~\cite{Aar09}, but  has remained largely unanswered. 
Two recent works~\cite{JK23arxiv,GMR23arxiv} consider this question; see more details in~\Cref{sec:related-work-anti-piracy}.

In the remainder of this introduction, we discuss in depth our approaches to these questions and our contributions. 
In \Cref{sec:intro-uncloneable}, we motivate our study of which problems have uncloneable proof systems by describing our new definition of what we call anti-piracy proof systems, and discuss methods to achieve it; we also compare with recent related work on uncloneable proofs~\cite{JK23arxiv,GMR23arxiv}. 
In \Cref{sec:intro-cloneable}, we discuss the other perspective of which problems have cloneable proof systems
by examining  complexity-theoretic implications of cloneability and showing that in the cloneable setting having multiple unentangled provers does not increase the computational power of $\QMA$.

\subsection{Anti-piracy proof systems} \label{sec:intro-uncloneable}

The first part of our central question asks about the existence of uncloneable proof systems. Towards answering this question, we contribute  a new definition that we call \emph{anti-piracy proof systems}; in \Cref{sec:defn-challenges}, we motivate this concept and in \Cref{sec:defn-uncloneable} we detail the definition. 
We also discuss how to achieve this definition in \Cref{sec:intro-anti-piracy-schemes}. 

\subsubsection{Challenges in defining uncloneable proof systems} \label{sec:defn-challenges}

As mentioned, the no-cloning principle tells us that unknown quantum states cannot be generically copied; hence it would be reasonable to expect that families of quantum states, such as proofs for $\QMA$,  could be shown to be uncloneable. 
However, such statement requires a careful approach. 
First, in the context of $\QMA$, any statement regarding the hardness of cloning quantum witnesses cannot be information theoretical: given the input to a problem in $\QMA$, one can find, in exponential time, a quantum proof that satisfies it. 
Secondly, as quantum witnesses are very structured, one would need to prove the computational uncloneability for such families of states. Examples of such families that appear in cryptography (but not proof systems) are the Wiesner states \cite{Wie83,MVW13} and the subspace/coset states~\cite{AC12,CV22}.

Towards understanding the  fundamentals of uncloneable proof systems, consider the following candidate for a one-message uncloneable proof system for a promise problem $A = (A_{yes},A_{no})$.

\paragraph{First attempt at a definition for uncloneable proof systems.} 
Given a quantum proof system satisfying completeness and soundness, we require that if $\ket{\psi_x}$ is a witness for $x \in A_{yes}$, then for any polynomial-time adversary $\mathcal{A}$, the probability that $\mathcal{A}(\ket{\psi_x}) = \ket{\psi_x} \ket{\psi_x}$ is negligible. 

While our first attempt seems reasonable, below, we describe below multiple shortcomings. 
(We keep the exposition to product states and equality for simplicity.)  

\paragraph{``Cloning a functionality''.} 
Consider an algorithm $\mathcal{A}$ such that $\mathcal{A}(\ket{\psi_x}) = \ket{\phi_A} \ket{\phi_B}$, where both $\ket{\phi_A}$ and $\ket{\phi_B}$ make $V$ accept $x$ with high probability but $\ket{\phi_A}\ket{\phi_B} \neq \ket{\psi_x}\ket{\psi_x}$. 
Formally speaking, while $\mathcal{A}$ did not ``clone'' the state $\ket{\psi_x}$, $\mathcal{A}$ is able to ``clone its functionality'', since $\ket{\phi_A} \ket{\phi_B}$ can be used to convince two independent verifiers that $x$ is a yes-instance. 
This situation is not ruled out by our first attempt.\footnote{To mitigate the above situation, one approach would be to describe a  verifier~$V$ such that for every $x$, there is a single $\ket{\psi_x}$ that makes $V$ accept $x$ with high probability, for instance, as in \emph{Unique} $\QMA$ ($\UQMA$)\cite{aharonov2022pursuit}). 
For such proof systems, ruling out cloning would also rule out the behaviour of $\mathcal{A}$ described above. 
However, looking ahead, our next shortcoming would be applicable to this situation.}

\paragraph{``Differing verifications and pirate-verifier collusions''.}
Consider an algorithm $\mathcal{A'}$ such that $\mathcal{A'}(\ket{\psi_x}) = \ket{\phi'_A} \ket{\phi'_B}$ and a verifier $V'\neq V$ that accepts $\ket{\phi'_A}$ and $\ket{\phi'_B}$ with high probability. 
As long as such a $V'$ leads to a sound proof system, we believe that such a situation should be disallowed by a correct definition; 
however our first attempt at a definition does not rule out such a scenario where alternative verification schemes exist, possibly in collusion with the adversarial attack. 
Two motivating scenarios follow. 

\begin{itemize}

\item \emph{Quantum-classical proofs:} 
Consider a proof systems with $\ket{\psi_x} = \ket{\psi_q}\ket{\psi_c}$, where $\ket{\psi_q}$ is an uncloneable quantum state that is independent of $A$ (e.g.~Wiesner states or subspace/coset states), and $\ket{\psi_c}$ is a cloneable state (for example a classical string) that is used as a classical witness (such as in $\NP$ or $\QCMA$). 
In this case, the proof system might be ``uncloneable'' for the honest verification which would first verify the quantum state, and proceed to verify the classical state only if the quantum state is deemed authentic. 
However, it is easy to make this proof system cloneable by considering a pirate that outputs copies of $\ket{\psi_c}$ to valid verifiers that bypass the quantum verification of $\ket{\psi_q}$ and only verify $\ket{\psi_c}$. 

\item \emph{Collusions between the pirate and verifier.} 
As a generalization of the above, we note that from the point of view of someone that wants to decide if  $x {\in} A_{yes}$, one obvious approach is to deviate from the honest behaviour $V$. 
In this situation, the pirate and the would-be verifier~$V'$ are in a collusion, since $V'$ could depend on the actions of the pirate. This situation is similar to what we find in practice  with pirated media: users of such media are satisfied as long as the end-result is the same or close enough; for this, bypassing the ``honest'' decoders that comply with the distributor's requirements is common, and the result is that the method to render the media differs in the intended and pirated setting. 
\end{itemize}

We conclude that our first attempt is insufficient and via the  situations described above, we have a better grasp at the subtleties that must be considered towards a robust definition of uncloneability in proof systems.

\subsubsection{Our definition: anti-piracy proof systems}
\label{sec:defn-uncloneable}
Our work contributes a formal definition of \emph{anti-piracy} proof systems, which,
we believe, captures a general sense of uncloneability in proof systems.  Since \emph{pirates} are a more powerful case than \emph{cloners}, our choice of terminology emphasises the fundamentally new perspective on uncloneablity in the context of proof systems:
\begin{quote}
    \textbf{Contribution 1:} A formal definition of \emph{anti-piracy} proof systems.
\end{quote}

As pointed out in \Cref{sec:defn-challenges}, an anti-piracy proof system must not only account for pirates that have access to a proof and attempt to modify it; it must also account for  deviations in the verification. To this end, we define a key notion of an \emph{admissible} verifier for a given promise problem~$A$:  
  
\begin{quote}
     A verifier $V$ is \emph{admissible} for~$A$ if there exists a prover $P$ such that the resulting proof system $(P,V)$ satisfies completeness and soundness for~$A$. 
\end{quote}
    
With this  notion in hand, towards defining anti-piracy proof system, we can quantify over ``reasonable'' verifiers: they want to be convinced in the case that there is a valid proof (there exists a prover that would satisfy completeness), but are not overly gullible (soundness must hold for all possible provers). For instance, this notion lets us exclude verifiers that always reject or always accept. We can now give an informal version of our main definition:

\begin{quote}
A proof system $(P,V)$ for $A$ is an anti-piracy proof system  if, for any $x \in A_{yes}$, any polynomial-time pirate having access to a proof for $x$ and that outputs two registers $(R_1,R_2)$, then for any pair of admissible verifiers $(V_1, V_2)$ for $A$, the probability that \emph{both} verifiers accept given registers $(R_1,R_2)$, respectively, is negligible. 
\end{quote}

One problem with the previous definition is that it requires the pirate to perform badly for {\em all} instances. This is a problem because even for hard problems in worst/average case, there could be a big fraction of the inputs for which it is easy to find a witness. As in the definition of witness hiding~\cite{feige1990witness}, the solution is to ask for a distribution over the inputs on which the pirate behaves poorly.

\paragraph{Further remarks.}
We now mention some additional remarks on our definition:
\begin{itemize}
\item The property of being anti-piracy related to a proof system and not a problem/language. Thus, it is conceivable that a language would admit both an anti-piracy proof system and another proof system whose proofs are perfectly cloneable. In particular, our work conjectures the possibility of anti-piracy proof systems for all problems in NP (see \Cref{sec:intro-anti-piracy-schemes}).

    \item For simplicity, we focus on \emph{non-interactive} (1-message) proof systems. 
     Moreover, we consider here only {\em public verification}. 
     If, instead, we had designated verifiers, where the prover and the verifier share private resources such as private keys or EPR pairs, then achieving such a definition would be straightforward. We stress that the pirate should have access to the same resources as the verifiers, and in particular, they can verify the proof.
See~\Cref{sec:openQuestions} for further discussion.
    \item Anti-piracy proof systems are only meaningful for hard problems and instances. In the absence of such hardness, an admissible verifier can disregard any input while ensuring completeness and soundness. We note that this is in stark contrast with zero-knowledge proof systems, where easy problems and instances trivially satisfy the zero-knowledge property. 
 However, as in  zero-knowledge, anti-piracy implies witness hiding \cite{feige1990witness}. This is because, in any proof system that is not witness hiding, a pirate should be able to extract a ``classical'' witness and given its nature, can then clone it.

\end{itemize}

\subsubsection{Schemes that achieve the definition}
\label{sec:intro-anti-piracy-schemes}
Our definition  for anti-piracy proof systems opens a completely new paradigm for uncloneability, that goes beyond the uncloneability of any specific information or functionality, focusing on the intrinsic usefulness of a proof and how quantum information could be used to limit this usefulness to multiple verifiers.

The obvious next question is: 
    \emph{Is it possible to achieve anti-piracy proof systems?}
Towards an answer, we note the intrinsic difficulty to the problem, due to the complexity of its definition: because we quantify over \emph{all} possible admissible verifiers (which implicitly quantifies over all possible provers), and then over all possible pirates,  achieving a proof is particularly unwieldy. 

We thus approach the question stepwise, first showing a conditional result. 

\begin{quote}
    \textbf{Contribution 2:} We provide an anti-piracy proof system for an oracle problem. Our oracle anti-piracy proof system makes use of quantum subspace states and subspace membership oracles, as well as a key result from Ben-David and Sattah, on the probability of finding strings in a subspace and their dual given these resources~\cite{ben2023quantum}.
\end{quote}

We present our anti-piracy proof systems relative to an oracle in \Cref{sec:uncloneable_oracle}.   Next, towards an unconditional result on anti-piracy proof systems, we give the following. 

\begin{quote}
    \textbf{Contribution 3:} We provide a candidate for uncloneable proof systems for $\NP$. This candidate builds on the idea of the uncloneable subspace state in Contribution 2, but instead of oracles for subspace membership, it uses cryptographic techniques such as obfuscation and non-interactive zero-knowledge proof systems.
\end{quote}

In \Cref{sec:uncloneable_np}, we explain the candidate and explore some challenges towards a  proof.

\subsubsection{Related work} 
\label{sec:related-work-anti-piracy}

We now mention related work on uncloneability in quantum information, comparing with our work. 

\paragraph{Uncloneable proofs.}

There are two recent results on uncloneable proofs~\cite{JK23arxiv,GMR23arxiv} and we now compare with our work.

In \cite{JK23arxiv}, the authors define the notion of uncloneable non-interactive zero-knowledge proof systems (NIZKs) which capture the standard notions of completeness, soundness and zero-knowledge, with the additional requirement that, given an honest proof in the protocol, no adversary can pass \emph{two} honest verifications of the proof system (here by honest verification, we mean the verification prescribed by the protocol). 
In particular, in their proposed protocols, in order to check whether the instance is positive or negative, the verifier only checks classical information. 
The uncloneable (quantum) piece of the proof is then appended to guarantee uncloneability. While such protocols find their applications elsewhere, we note that they do not satisfy our definition of anti-piracy, since the uncloneable part could be easily removed and the proof would satisfy an admissible verification. 

In \cite{GMR23arxiv}, the authors present two definitions of uncloneable NIZKs. 
The first one is similar to~\cite{JK23arxiv} that we have discussed above. 
In a second definition, called \emph{strong} uncloneable proofs, they
require that if the adversary can make two (not necessarily honest) verifiers accept, then we can extract one of the witnesses from the adversary (without the original proof). 
While \cite{GMR23arxiv} demonstrate impossibility results in this model, their results are not applicable to our scenario. The main reason is that \cite{GMR23arxiv} allows interaction between a would-be verifier and an adversary (called a ``right-receiver'' and ``\textsf{MiM}'' in their work), meaning that a first would-be verifier can return a message to the adversary, who then attempts to convince a second would-be verifier.  In contrast, our definition closely models a one-message proof system, enforcing a \emph{simulatenous} evaluation by possibly malicious verifiers.

\paragraph{Copy-Protecting information.} 
Wiesner's visionary ideas in the late 1960s (but not published until much later~\cite{Wie83}) 
are the foundations of much of present-day quantum cryptography, including quantum key distribution~\cite{BB84,Eke91}, quantum money~\cite{AFG+12}, uncloneable encryption~\cite{BL20}, certified deletion~\cite{BI20} and more (for a survey, see~\cite{BS16}). 
These works, however are fundamentally different from uncloneability proofs because the primitives are used for encoding strings (as in encryption), or verify authenticity only (as in money). 

\paragraph{Copy-protecting software.} 
In groundbreaking work~\cite{Aar09} pushed beyond the paradigm of uncloneable information, studying for the first time the idea of uncloneable \emph{functionalities} by defining \emph{quantum copy-protection}, and proposing related schemes. 
However, this work left many open questions; among them, the question of building \emph{uncloneable quantum software} from standard cryptographic assumptions; follow-up work solved this problem under multiple angles~\cite{ALP21,BJL+21,CMP24}. 
In uncloneable software, there is an explicit input-output behaviour with respect to the evaluator, and this is a key element for the corresponding security game (which relates to challenge distributions, and the limited ability of various parties to correctly predict the outcome of a computation, given an input). 
In contrast, in anti-piracy proof systems, there is no challenge distribution, which makes the task of even defining such a scheme quite difficult. 

\paragraph{Uncloneable advice.}

Continuing on the line of work of studying uncloneability and its use, very recently \cite{BKL23arxiv} studied uncloneability of quantum advice. They showed that there exist a non-uniform language whose advice states are uncloneable. 
Following the framework of quantum copy-protection, their definition is given with respect to a challenge distribution.  
However, in contrast with quantum copy-protection (and more along the lines of anti-piracy quantum proofs), the uncloneable advice state is unkeyed. 

\subsection{Cloneable proof systems and their implications}
\label{sec:intro-cloneable}
We now switch to discussing cloneable proof systems and their implications to complexity theory. 
A major pressing question in quantum complexity theory, first asked by Aharonov and Naveh~\cite{AN02arxiv}, is to understand the power of quantum proofs over classical proofs. 
A notorious example is understanding the relationship between $\QMA$ and its classical proof counterpart, $\QCMA$. 
Progress towards answering this question appears in~\cite{Wat00,AK07,KP14b,FK15arxiv, Kla17, li2023classical, natarajan2024distribution, ben2024oracle}, however the main question still remains. 

One method to  separate $\QCMA$ and $\QMA$ is to find a complexity class that sits between them. 
For example, in \cite{GKS16}, they consider simpler quantum states, called subset states, as witnesses, but they show that this simplification is still strong enough to encompass $\QMA$.  
Following \cite{NZ24}, in this work we consider a different aspect of quantum states which is cloneability. 
Since classical states are perfectly cloneable, classical witnesses are perfectly cloneable. 
(There are many interesting complexity classes which are defined with classical witnesses such as $\NP$, its randomized variant $\MA$, and $\QCMA$.) 
We aim to understand the consequences of general {\em quantum witnesses} that \emph{can} be copied, which would make them closer to classical in a sense. 
 
More concretely, \cite{NZ24} defines the complexity class $\CQMA$ where in the completeness case, the prover can send a state which is  \emph{copiable}, up to some fidelity parameter $f$. Naturally, $\CQMA$ is a subset of $\QMA$ and a superset of $\QCMA$. 
Clearly, if $\CQMA=\QMA$ (or the stronger statement $\QCMA = \QMA$), then every problem in $\QMA$ has perfectly cloneable proof. 
On the other hand, any separation (e.g. oracle) between $\CQMA$ and $\QMA$ or $\QCMA$ and $\CQMA$ would give the same type of separation between $\QMA$ and $\QCMA$ as well. 
\cite{NZ24} gives a \emph{quantum oracle} separation between $\CQMA$ and $\QMA$ reproving the result of \cite{AK07} in the cloneable setting. Moreover if $\QMA\neq\QCMA$, then it is intriguing where $\CQMA$ lies between them, and which problems outside $\QCMA$ admit cloneable proofs~\cite{Aar16}.

In this work, we push forward the consequences of having cloneable proofs.

\subsubsection{Summary of Contributions: Implications of Cloneable Proof Systems} 

In \Cref{section:complexity_of_piracy}, we build on the work of \cite{NZ24}, which defines and studies the complexity class $\CQMA$, extending it to unentangled provers and studying its behaviour. 

\begin{quote}
    \textbf{Contribution 4:} 
    We show that in the uncloneable regime, there is no power of unentanglement. 
    More precisely, we define the analogous class $\CQMAk$, referring to $k$ unentangled provers, and show that $\CQMAk=\CQMA$. 
\end{quote} 
\quad

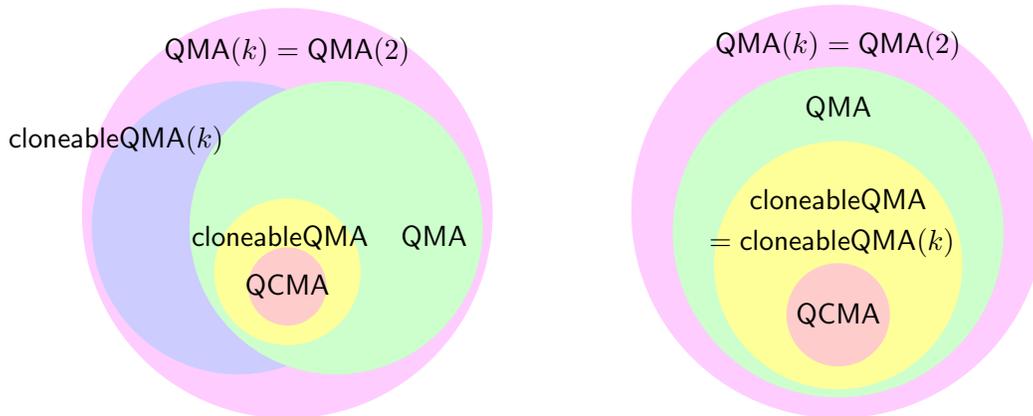
\begin{figure}[h]
    \centering
    \begin{tabular}{ccc}
        \begin{tikzpicture}[scale=0.65]
            \definecolor{setColorA}{RGB}{255, 204, 204}

\newcommand{\blue}[1]{\textcolor{blue}{#1}}
\newcommand{\red}[1]{\textcolor{red}{#1}} 
\newcommand{\gray}[1]{\textcolor{gray}{#1}}
\newcommand{\teal}[1]{\textcolor{teal}{#1}}

\newcommand{\demicolon}{%
        \textcolor{red}{
	\rotatebox[x=1.391pt,y=3.777pt]{-120}{;}
	\hspace{-9.2pt}
	;
	\hspace{-9.2pt}
	\rotatebox[x=1.391pt,y=3.777pt]{120}{;}
}}

\newcommand{\note}[1]{\blue{[* #1 *]}} 
\newcommand{\snote}[1]{\blue{[\demicolon Jamie: #1 \demicolon]}} 
\newcommand{\jnote}[1]{\blue{[\demicolon Jamie: #1 \demicolon]}} 
\newcommand{\jsnote}[1]{\blue{[\demicolon Jamie: #1 \demicolon]}}   
\newcommand{\pnote}[1]{\teal{[*** Supartha: #1 ***]}} 
\newcommand{\bnote}[1]{\blue{[*** Anne: #1 ***]}} 
\newcommand{\agnote}[1]{\blue{[*** Alex: #1 ***]}}
            \definecolor{setColorB}{RGB}{255, 255, 153}
            \definecolor{setColorD}{RGB}{204, 255, 204}
            \definecolor{setColorC}{RGB}{204, 204, 255}
            \definecolor{setColorE}{RGB}{255, 204, 255}
            \fill[setColorE, opacity=0.3] (0,1.5) circle (4.2);
            \fill[setColorC, opacity=0.5] (-1,1.2) circle (3);
            \fill[setColorD, opacity=0.5] (1,1.2) circle (3);
            \fill[setColorB, opacity=0.5] (0,0.3) circle (1.5);
            \fill[setColorA, opacity=0.5] (0,0) circle (0.8);
            \node at (0, 4.8) {$\QMA(k) = \QMA(2)$};
            \node at (-3.5, 3) {$\CQMAk$};
            \node at (3, 1) {$\QMA$};
            \node at (-0.15, 1) {$\CQMA$};
            \node at (0, 0) {$\QCMA$};
        \end{tikzpicture}
        & \hspace{1cm} &
        \begin{tikzpicture}[scale=0.55]
            \definecolor{setColorA}{RGB}{255, 204, 204}
            \definecolor{setColorB}{RGB}{255, 255, 153}
            \definecolor{setColorC}{RGB}{204, 255, 204}
            \definecolor{setColorD}{RGB}{255, 255, 153}%
            \definecolor{setColorE}{RGB}{255, 204, 255}
            \fill[setColorE, opacity=0.3] (0,2.5) circle (5);
            \fill[setColorC, opacity=0.5] (0,2) circle (4);
            \fill[setColorD, opacity=0.5] (0,1.2) circle (3);
            \fill[setColorB, opacity=0.0] (0,0.3) circle (2);
            \fill[setColorA, opacity=0.5] (0,0) circle (1.25);
            \node at (0, 6.5) {$\QMA(k) = \QMA(2)$};
            \node at (0, 2.75) {$\CQMA$};
            \node at (0, 5) {$\QMA$};
            \node at (-0.15, 1.75) {$=\CQMAk$};
            \node at (0, 0) {$\QCMA$};
        \end{tikzpicture} \\
    \end{tabular}  
    \caption{The $\QMA$ landscape with respect to cloneable proof systems. 
    (Left) Immediate/known containments. 
    (Right) Containments following our work. %
    }   
\label{fig:complexity_side_by_side}
\label{fig:complexity1}\label{fig:complexity_side_by_side_intro}
\end{figure}

 \Cref{fig:complexity_side_by_side_intro} depicts \CQMA with respect to other variants of $\QMA$. 
One could ask if any of the containments of $\QCMA \subseteq \CQMA \subseteq \QMA \subseteq \QMAtwo$ are strict or not. 
Resolving this question is of great interest. 
For one reason, any strict containment between \CQMA and any of these other classes would separate $\QCMA$ from $\QMAtwo$, or possibly from $\QMA$, a famous open problem in complexity theory. 
On the other hand, proving $\CQMA = \QMA$ implies that $\QMA$ witnesses can, in general, forfeit the quantum feature of being uncloneable, yielding great insight towards the nature of quantum proofs. 
Lastly, proving $\CQMAtwo = \QMAtwo$ would imply $\QMAtwo = \QMA$, settling a longstanding open problem.

\subsection{Open Questions and Future Work}
\label{sec:openQuestions}

After many recent works on uncloneability of quantum states and their applications to cryptography, we are still far from a full understanding. Below is a non-exhaustive list of some of the main open questions that are directly related to our work:

\begin{itemize}
    \item As we discussed before, the anti-piracy proof system is fundamentally different in nature than the other uncloneable proof systems. Towards its further exploration, can we design anti-piracy proof systems for a more natural problem in $\NP$ (or more generally for $\QMA$)? Our work provides a candidate construction for an anti-piracy proof system for all of~$\NP$. Can we prove that it satisfies our definition? 

    \item Uncloneable primitives often reduce to one another. It would be interesting to know if the definition of an anti-piracy proof system can imply or can be implied by other uncloneable cryptography primitives? Also, like zero-knowledge, anti-piracy implies witness hiding. Are there further connections of anti-piracy with zero knowledge proof systems?

\item We leave for future work the study of uncloneable proof systems where  interaction is considered. When multiple rounds are possible, a solution involving first the verification of shared entanglement, followed by teleportation seems plausible, while in a model without interaction, recent techniques related to \emph{monogamy-of-entanglement} would seem relevant~\cite{BC23arxiv}.

\item  We note that two open problems in complexity theory can be resolved via cloneable proof variants. 
Contribution 4 above says that if the complexity class $\QMAtwo$ equals its cloneable counterpart, then this resolves the $\QMA$ vs.~$\QMAtwo$ problem. 
On the other hand, if we have $\QMA \neq \CQMA$, then this resolves the $\QMA$ vs.~$\QCMA$ problem (in the negative). 
Thus, we leave as an open problem making progress in these important questions by studying the power of (un)cloneability.  

\item Our definition of $\CQMAk$ leaves room for variations in the setup and in the parameters. We discuss some of the variations in \Cref{rem:cloneable} and \Cref{sec:remarks-cloneable}, and we leave the extension of our results to these settings as future work.

\end{itemize}

\subsection{Organization}

Our work on anti-piracy proof systems is in \Cref{sec:Uncloneable-Proofs}, while the complexity-theoretic results of cloneability in QMA proof systems are covered in \Cref{section:complexity_of_piracy}. 
We defer the preliminaries to \Cref{Section:prilim}.

\subsection*{Acknowledgements}  

We would like to thank Barak Nehoran for useful discussions and comments.  
We acknowledge support of the Natural Sciences and Engineering Research Council of Canada (NSERC). This work was supported by the Air Force Office of Scientific Research under award number FA9550-20-1-0375.
S.P.\ is supported by US Department of Energy (grant no DE-SC0023179) and partially supported by US National Science
Foundation (award no 1954311). S.P.\ thanks the University of Ottawa, where a portion of this work was completed. This work was done in part while ABG was visiting the Simons Institute for the Theory of Computing. ABG is funded by ANR JCJC TCS-NISQ ANR-22-CE47-0004 and by the PEPR integrated
project EPiQ ANR-22-PETQ-0007 part of Plan France 2030.

\section{Towards anti-piracy proof systems}
\label{sec:Uncloneable-Proofs}
In this section, we discuss the possibility of anti-piracy proof systems.  We start by providing the definition of anti-piracy proofs in \Cref{sec:uncloneable_definition}. Then, in \Cref{sec:uncloneable_oracle}, we show an anti-piracy proof for an oracle problem related to subspace states. Finally, in \Cref{sec:uncloneable_np}, we provide a candidate for anti-piracy proof systems for all of $\NP$. 

\subsection{Definition}
\label{sec:uncloneable_definition}

In order to define anti-piracy proof systems, we start with defining admissible verifiers and non-interactive proof systems.

\begin{definition}[Admissible verifier and non-interactive proof system]\label{def:admissible}
Let $A = (A_{yes},A_{no})$ be a promise problem. We say that a quantum polynomial-time algorithm $V$ is an admissible verifier for~$A$ if:
\begin{itemize}
    \item[] \textbf{completeness:} if $x \in A_{yes}$, there exists a proof $\ket{\psi}$ such that $\Pr[V(x,\ket{\psi}) = 1] \geq 1 -\negl(|x|)$.
    \item[] \textbf{soundness:} if $x \in A_{no}$, for all proofs $\ket{\psi}$ such that $\Pr[V(x,\ket{\psi}) = 1] \leq \negl(|x|)$.
\end{itemize}

\noindent Moreover, for a possibly unbounded prover $P$, $(P,V)$ is a non-interactive proof system for $A$ if for all $x \in A_{yes}$, $\Pr[V(x,P(x)) = 1] \geq 1 -\negl(|x|)$.
\end{definition}

We notice that we define these notions with negligible completeness and soundness error since these are the relevant parameters in the cryptographic setting, which is our motivation in this section.

As we discussed in the introduction, our notion of anti-piracy proofs only makes sense in average case. For that, we define now a generator of instances as in \cite{feige1990witness}.

\begin{definition}
 A generator $G$ for some promise problem $A = (A_{yes},A_{no})$ is a  randomized polynomial-time algorithm  that on input $1^n$ outputs an $x \in A_{yes}$ such that $|x|  = n$.
\end{definition}

We define now the notion of anti-piracy proofs for a generator $G$.

\begin{definition}[Anti-piracy proof for generator $G$]
\label{def:uncloneable_proof_generator}
Let $A = (A_{yes},A_{no})$ be a promise problem and $G$ be a generator for $A$.
We say that a non-interactive proof system  $(P, V)$ 
for $A$ is anti-piracy for $G$, 
if for every $\poly(|x|,\lambda)$-bounded adversary $\calA$ and admissible verifiers $V_1$ and $V_2$ for $A$, we have that: 
\begin{align}
    Pr_{x \sim G(1^n)}\left[
      V^x_1 \otimes V^x_2 (\tilde\psi_{A,B}) = (1,1)
 \quad \bigg| \quad
\substack{
  \ket{\psi}
  \leftarrow P(x) \\
  \tilde{\psi}_{A,B} \leftarrow
  \mathcal{A}(x,\ket{\psi})
}
      \right] \leq \negl(\lambda),
\end{align}
where the probability space is also over the measurement outcomes of $V_1$ and $V_2$.
\end{definition}

Finally, we can define anti-piracy proof systems for a promise problem.

\begin{definition}[Anti-piracy proof for promise problem $A$]
\label{def:uncloneable_proof}
Let $A = (A_{yes},A_{no})$ be a promise problem.
We say that a proof system $(P, V)$ 
for $A$ is anti-piracy if there is a generator $G$ such that  $(P, V)$ is anti-piracy for $G$.
\end{definition}

\begin{remark}
    We note that in \Cref{def:uncloneable_proof}, the proof system could be constructed using oracles or a common reference string (CRS). In this case, we consider that the adversary $\calA$ or any admissible verifier for the problem has access to the same resources as the honest verifier. Notice that if the honest prover and verifier pair has access to private resources, such as private keys or EPR pairs, this would lead to a trivial anti-piracy proof system but with a designated verifier. In this work, we focus on proof systems with public verification.
\end{remark}

\subsection{Anti-piracy proofs with subspace oracles}
\label{sec:uncloneable_oracle}
In this section, we show an oracle problem for which there is an anti-piracy proof. We start by defining the problem.

\begin{definition}\label{def:uncloneable_language}
Let $A_n,B_n \in \mathbb{F}_2^n$ be subspaces such that $B \subseteq A^\perp$ and $\dim(A) = \dim(B) = n/2$, or $\dim(A) = n/2$  and $\dim(B) = n/4$, or $\dim(A) = n/4$  and $\dim(B) = n/2$.  Let ${A} = (A_1,...)$ and ${B} = (B_1,...)$.  We define a language $L_{{A},{B}}$ as 
    \begin{enumerate}
    \item $0^n \in L_{{A},{B}}$ iff $\dim(A_n) = \dim(B_n) = n/2$
     \item $x \not\in L_{{A},{B}}$ for all $x \ne 0^n$
    \end{enumerate}
 
    For each $n$, we will define the pair of oracles $O^{A_n}$ and $O^{B_n}$ that check membership in $A_n$ and~$B_n$, respectively.
\end{definition}

If we want to check that the support of an $n$-qubit quantum state $\ket{\psi} = \sum_x \alpha_x \ket{x}$ is in $A_n$, we make an oracle call $O^{A_n}$ on $\ket{\psi}\ket{0}$, which results in 
    $\sum_{x \in A_n} \alpha_x \ket{x}{1} + \sum_{x \not\in A_n} \alpha_x \ket{x}\ket{0}$. We can then simply measure the last qubit and if the outcome is $1$, then we accept; or if the outcome is $0$, then we reject. A similar procedure can be defined for $B_n$.

\begin{figure}[h]
\fbox{\begin{minipage}{0.98\textwidth} 
  \textbf{Input:} input $x \in \{0,1\}^n$ and $n$-qubit proof $\ket\psi$
  \begin{enumerate}
    \item If $x \ne 0^n$, reject.
    \item Check if each element of $\ket\psi$ belongs to the subspace $A_n$, reject if not
    \item Apply $H^{\otimes n}$ on the state
    \item Check if each element of the new state belongs to the subspace      $B_n$, reject if not
    \item Accept
  \end{enumerate}
  \end{minipage}}
  \caption{Admissible verification $V^*$ for $L_{A,B}$}
  \label{fig:honest-verification}
\end{figure}

\begin{lemma}\label{lem:honest-verification}
For any $A,B$ satisfying the properties in \Cref{def:uncloneable_language}, let $V^*$ be defined as in \Cref{fig:honest-verification}, and  $P^*$ be the such that on input $x = 0^n$, outputs $\ket{A_n}$. Then $(V^*, P^*)$ is a non-interactive proof system for~$L_{A,B}$.\looseness=-1
\end{lemma}
\begin{proof}
For completeness, we have that for all $n$ such that $0^n \in L$, $V^{\mathcal{O}^A,\mathcal{O}^{B}}$ accepts the proof $\ket{A_n}$ with probability $1$: the tests in step $2$ and $4$ never reject since $H^{\otimes}\ket{A} = \ket{A^\perp} = \ket{B}$.

To prove soundness, the case where $x \ne 0^n$ is trivial. Let us now consider the case where $x = 0^n \not\in L$. Let us assume that $\dim(A) = n/2$ and $\dim(B) = n/4$ (the other case follows analogously).

 Let us assume that the first test passes (if it does not pass, the verifier
  rejects as desired). We denote state of the verifier after the first test as
  $\ket{\phi} = \sum_{x \in A} \alpha_x \ket{x}$, and we use the fact that the first test has
  passed, so the support of $\ket{\phi}$ is a subset of $A$.

  When $V$ applies $H^{\otimes n}$ to $\ket{\phi}$, the resulting state is
  \begin{align}
     \frac{1}{2^{n/2}}\sum_{y \in \{0,1\}^n} \sum_{x \in A} (-1)^{x \cdot
     y}\alpha_x \ket{y},
  \end{align}
  and the probability of measuring $y \in B$ is
  \begin{align}
    &\sum_{y \in B} 
     \left|\frac{1}{2^{n/2}}\sum_{x \in A} \alpha_x \right|^2 \\
    &
    \leq \frac{1}{2^{n}} \sum_{y \in B} 
     \left|2^{n/2} 
     \frac{1}{2^{n/4}}\right|^2 \\
    &
    = \frac{1}{2^{n/4}}.
  \end{align}
  Thus, the second test passes with negligible probability.
\end{proof}

In order to prove the anti-piracy property, we first prove a result on the queries made by an admissible verifier for the language $L_{A,B}$.

\begin{lemma}\label{lem:intercept-query}
  Let $V^{\mathcal{O}^{A_n},\mathcal{O}^{B_n}}$ 
  be an admissible verifier for
  $L_{A,B}$. Let us also fix some $n$ and some polynomial $p$ such that $0^n \in L_{A,B}$ and some state $\ket{\psi}$ such that $V^{\mathcal{O}^{A_n},\mathcal{O}^{B_n}}$ accepts on $\ket{\psi}$ with probability at least $\frac{1}{p(n)}$.
  
  Then there exists a polynomial $q$ such that for every
  subspace $B_n' \subseteq A_n^\perp$ of dimension $n/4$, there
  exists one query to $\mathcal{O}^{B_n}$ that has $\frac{1}{q(n)}$ mass on elements
  in $B_n \setminus B_n'$. Moreover, for every subspace $A_n' \subseteq A$ of dimension
  $n/4$, there
  exists one query to $\mathcal{O}^{A_n}$ that has non-negligible mass on elements
  in $A_n \setminus A_n'$.
\end{lemma}
\begin{proof}
  Let us suppose that the statement is false. In this case, we have that there is a negligible function $\eps$ such that  every query to $\mathcal{O}^{B_n}$ has $\eps(n)$ mass on elements
  in $B_n \setminus B_n'$.
  
  We will show that we can replace the
  oracle $B_n$ with an oracle $B_n' \subseteq A_n^\perp$ of dimension $n/4$, and that $V^{\mathcal{O}^{A_n},\mathcal{O}^{B_n'}}$ accepts $0^{n}$ with
  probability $\frac{1}{p(n)} - \eps(n)$, contradicting the assumption that $V^{\mathcal{O}^{A_n},\mathcal{O}^{B_n'}}$ is an admissible verifier of $L_{A,B'}$.

  Without loss of generality, we assume that $V^{\mathcal{O}^{A_n},\mathcal{O}^{B_n}}$  can be decomposed as 
  \begin{align} U_{T+1} \mathcal{O}^{B_n} U^B_{T} ... \mathcal{O}^{B_n} U^B_{1}{\mathcal{O}^{A_n}}U^A_{1}\mathcal{O}^{B_n} U^B_0\mathcal{O}^{A_n}U^A_0,
  \end{align}
where $U^A_i$ and $U^B_i$ are arbitrary unitaries.

Let us define 
  \begin{align*} \ket{\psi_i} = U^B_{i} ... \mathcal{O}^{B_n}  U^B_0\mathcal{O}^{A_n}U^A_0\ket{x}\ket{\psi}\ket{0} \quad \text{ and } \quad \ket{\phi_i} = U^B_{i} ... \mathcal{O}^{B_n'} U^B_{1}{\mathcal{O}^{A_n}}U^A_{1}\mathcal{O}^{B_n'} U^B_0\mathcal{O}^{A_n}U^A_0\ket{\psi}\ket{0}.
  \end{align*}

  We will show by induction that there exists a negligible function $\eps(n)$ such that for every $0 \leq i \leq T$
  \begin{align} \label{eq:induction-step-unclonable}
      \norm{\ket{\psi_i} - \ket{\phi_i}} \leq i\eps(n).
  \end{align}

  Notice that from \Cref{eq:induction-step-unclonable}, we can prove our statement since $\ket{\phi_T}$ is the state of the verifier $V^{\mathcal{O}^{A_n},\mathcal{O}^{B_n'}}$ on input $\ket{\psi}\ket{0}$. In this case, the acceptance probability of the no-instance $0^n$ of $L_{A,B'}$ would be $\frac{2}{3} -\negl(n)$ and therefore the verification is not sound, which is a contradiction.

We now prove \Cref{eq:induction-step-unclonable}.
  For the base case $i = 0$, we have that $\ket{\psi_0} = \ket{\phi_0}$, and therefore the statement holds.

  Let us now assume that \Cref{eq:induction-step-unclonable} holds for some $i$. Then we have that
  \begin{align*} 
      &\norm{\ket{\psi_{i+1}} - \ket{\phi_{i+1}}} \\
      &= \norm{U^B_{i+1}{\mathcal{O}^{A_n}}U^A_{i+1}\mathcal{O}^{B_n}\ket{\psi_{i}} - U^B_{i+1}{\mathcal{O}^{A_n}}U^A_{i+1}\mathcal{O}^{B_n'}\ket{\phi_{i}}} \\
      &=  \norm{\mathcal{O}^{B_n}\ket{\psi_{i}} - \mathcal{O}^{B_n'}\ket{\phi_{i}}} \\
      &\leq \norm{\mathcal{O}^{B_n}\ket{\psi_{i}} - \mathcal{O}^{B_n'}\ket{\psi_{i}}} +
      \norm{\mathcal{O}^{B_n'}\ket{\psi_{i}} - \mathcal{O}^{B_n'}\ket{\phi_{i}}} \\
      &\leq \norm{\mathcal{O}^{B_n}\ket{\psi_{i}} - \mathcal{O}^{B_n'}\ket{\psi_{i}}} +
      i\eps(n) \\
      &\leq \eps(n) +  i\eps(n) \\
      &= (i+1)\eps(n),
  \end{align*}
  where the second equality follows by unitary invariance of the norm, the first inequality follows by the triangle inequality, and the second inequality follows from the induction hypothesis. We argue now the third inequality.
Notice that by the assumption, all queries to $\mathcal{O}^{B_n}$ have $\eps(n)$ mass on $B_n \setminus B_n'$, i.e.,  $\sum_{x\in B\setminus B_n', b, y} |a_{x,b,y}|^2 \leq \eps(n)$. In this case, 
$\norm{\mathcal{O}^{B_n}\ket{\psi_i} - \mathcal{O}^{B_n'}\ket{\psi_i}} \leq \eps(n)$. 
\end{proof}

We now state a key technical lemma from Ben-David and Sattath.

\begin{lemma}[Theorem 16 of~\cite{ben2023quantum}]
  \label{lem:BDS}

  Let A be a uniformly random subspace, and let  $\epsilon > 0$ be such that
  $1/\epsilon = o (2^{n/2})$. Given one copy of $\ket{A}$ and a quantum membership oracle 
 $A$ and $A^\perp$, a
  counterfeiter needs $\Omega(\sqrt{\epsilon}2^{n/4})$ queries to output a pair
  $(a,b) \in A \times A^\perp$.
\end{lemma}

\begin{theorem}
Let $V^*$ and $P^*$ be defined as \Cref{lem:honest-verification}. Then $(V^*, P^*)$ is an anti-piracy proof system for~$L_{A,B}$.
\end{theorem}
\begin{proof}
We will consider a generator $G$ such that $G(1^n)$ outputs $0^n$, if $0^n \in L_{A,B}$ or $\bot$ otherwise. 

Let $n$ be such $0^n \in L_{A,B}$ and $\calA$ be an adversary that receives $\psi$ created by the prover, has
oracle access to $\calO$, and outputs $\tilde \psi_{A,B}$.
Let $\tilde{V}_1^{\mathcal{O}^{A_n},\mathcal{O}^{B_n}}$ and $\tilde{V}_2^{\mathcal{O}^{A_n},\mathcal{O}^{B_n}}$ be two admissible verifiers for  $L_{A,B}$ with the input $0^n$ hardcoded. We need to show that 
\begin{align}
    Pr\left[
      \tilde{V}_1^{\mathcal{O}^{A_n},\mathcal{O}^{B_n}} \otimes \tilde{V}_2^{\mathcal{O}^{A_n},\mathcal{O}^{B_n}} (\tilde\psi_{A,B}) = (1,1)
 \quad \bigg| \quad
\substack{
  \mathcal{O}, \psi
  \leftarrow P \\
  \tilde{\psi}_{A,B} \leftarrow
  \mathcal{A}^{\mathcal{O}^{A_n},\mathcal{O}^{B_n}}(\psi)
}
      \right] \leq \negl(\lambda).
\end{align}

  Let $M$ be the POVM where we pick one of the queries uniformly at random,
  measure the input register and accept iff the measurement outcome is a vector
  in the subspace.

   By \Cref{lem:intercept-query}, we have that
\begin{align}
 & Pr\left[
    M \otimes \tilde{V}_2^{\mathcal{O}} (\tilde\psi_{A,B}) = (1,1)
 \quad \bigg| \quad
\substack{
  \mathcal{O}, \psi
  \leftarrow P \\
  \tilde{\psi}_{A,B} \leftarrow
  \mathcal{A}^{\mathcal{O}}(\psi)
}
      \right]
    \\ &\geq
 \frac{1}{p(n)} \cdot Pr\left[
    \tilde{V}_1^{\mathcal{O}} \otimes \tilde{V}_2^{\mathcal{O}} (\tilde\psi_{A,B}) = (1,1)
 \quad \bigg| \quad
\substack{
  \mathcal{O}, \psi
  \leftarrow P \\
  \tilde{\psi}_{A,B} \leftarrow
  \mathcal{A}^{\mathcal{O}}(\psi)
}
      \right] . \end{align}

  Notice that since $M$ and $\tilde{V}_2^{\mathcal{O}}$ commute, we can use
  \Cref{lem:intercept-query} again and we have that
\begin{align}
 & Pr\left[
      M \otimes M (\tilde\psi_{A,B}) = (1,1)
 \quad \bigg| \quad
\substack{
  \mathcal{O}, \psi
  \leftarrow P \\
  \tilde{\psi}_{A,B} \leftarrow
  \mathcal{A}^{\mathcal{O}}(\psi)
}
      \right] \\
  & \geq
\frac{1}{p(n)}  Pr\left[
    M \otimes \tilde{V}_2^{\mathcal{O}} (\tilde\psi_{A,B}) = (1,1)
 \quad \bigg| \quad
\substack{
  \mathcal{O}, \psi
  \leftarrow P \\
  \tilde{\psi}_{A,B} \leftarrow
  \mathcal{A}^{\mathcal{O}}(\psi)
}
      \right] 
    \\ 
    & \geq 
     \frac{1}{p(n)^2} \cdot Pr\left[
    \tilde{V}_1^{\mathcal{O}} \otimes \tilde{V}_2^{\mathcal{O}} (\tilde\psi_{A,B}) = (1,1)
 \quad \bigg| \quad
\substack{
  \mathcal{O}, \psi
  \leftarrow P \\
  \tilde{\psi}_{A,B} \leftarrow
  \mathcal{A}^{\mathcal{O}}(\psi)
}
      \right]
 . \end{align}

  By \Cref{lem:BDS}, we have that 
\begin{align}
  Pr\left[
      M \otimes M (\tilde\psi_{A,B}) = (1,1)
 \quad \bigg| \quad
\substack{
  \mathcal{O}, \psi
  \leftarrow P \\
  \tilde{\psi}_{A,B} \leftarrow
  \mathcal{A}^{\mathcal{O}}(\psi)
}
      \right] 
\leq \negl(\lambda), \end{align}
and therefore 
\begin{align}
Pr\left[
    \tilde{V}_1^{\mathcal{O}} \otimes \tilde{V}_2^{\mathcal{O}} (\tilde\psi_{A,B}) = (1,1)
 \quad \bigg| \quad
\substack{
  \mathcal{O}, \psi
  \leftarrow P \\
  \tilde{\psi}_{A,B} \leftarrow
  \mathcal{A}^{\mathcal{O}}(\psi)
}
      \right] \leq \negl(n),
\end{align}
proving our statement.
\end{proof}

\subsection{Candidates for anti-piracy proofs for $\NP$}
\label{sec:uncloneable_np}

We propose here a candidate scheme for anti-piracy proofs beyond the oracle model, for languages in $\NP$. The main idea is that the quantum proof consists in three parts:
\begin{enumerate}
    \item some auxiliary information provided by the prover;
    \item a potentially cloneable proof $\pi$ that shows that either the input is in the language or the auxiliary information has some {\em bad properties}, and, crucially, $\pi$ hides which is the case; 
    \item an uncloneable quantum state that certifies that the auxiliary information does {\em not} have bad properties.
\end{enumerate}

In order to implement such elements, we will make use of cryptographic primitives that we define here informally.

The first object that we will need is non-interactive zero-knowledge proofs (NIZK). In particular, we will consider NIZKs in the common reference string (CRS) model. Here, we have a prover $P$ and a verifier $V$, which both have access to a trusted reference string $\textsf{crs}$ that is sampled from a pre-defined distribution. The properties that we want from the proof system is:
\begin{itemize}
    \item[] \textbf{completeness:} If $x \in L$, then $V(x,P(x,\textsf{crs}),\textsf{crs}) = 1$;
    \item[] \textbf{soundness:} If $x \not\in L$, then for every $\tilde{P}$, $Pr[V(x,\tilde P(x,\textsf{crs}),\textsf{crs}) = 1] = \negl(|x|)$;
    \item[] \textbf{computational zero-knowledge:} There exists a polynomial-time simulator $\mathcal{S}$ such that
    $(\textsf{crs},P(x,\textsf{crs}))$ is computationally indistinguishable from $\mathcal{S}(x)$.
\end{itemize}

The second cryptographic notion that we need is indistinguishability obfuscation (iO). Here, we say that a polynomial-time algorithm $\mathcal{O}$ is an iO for a circuit class $\mathcal{C}$ if:
\begin{itemize}
    \item For any circuit $C \in \mathcal{C}$, $\mathcal{O}(C)$ computes the same function as $C$ with at most a polynomial slowdown;
    \item For any two circuits $C,C' \in \mathcal{C}$ that compute the same function, $\mathcal{O}(C)$ is computationally indistinguishable from $\mathcal{O}(C')$.
\end{itemize}

We assume here an iO scheme $\mathcal{O}$ for circuits $C_A$, where $A \subseteq \mathbb{F}_2^n$, that compute the membership of $A$. \cite{Zha19,Zha21} showed that given $\ket{A}, \mathcal{O}(C_A), \mathcal{O}(C_{A^\perp})$, where $dim(A) = n/2$, no polynomial-time algorithm can output $\ket{A}^{\otimes 2}$ with non-negligible probability.

Using the recipe described in steps $1$ to $3$ above, the auxiliary information in our proposal consists of obfuscations of a subspace state and its dual, both of dimension $n/2$. Then, the bad property is that one of the subspaces has dimension $n/2$ and the second one dimension $n/4$. Finally, the subspace state is the uncloneable quantum state that certifies that the dimensions are correct. Using these tools, our candidate is given in \Cref{fig:candidate-proof-system}.

\begin{figure}[h]
\fbox{\begin{minipage}{0.98\textwidth}

Fix a language $L \in \NP$. We define the language $L_{sub}$ as

\begin{align}
  L_{sub} = \{ (x,O_1,O_2) : & \exists B, C \subseteq \mathbb{F}_2^n \text{ s.t. 
  }\\ & O_1 = \mathcal{O}(B)  
   \wedge O_2 = \mathcal{O}(C) 
  \wedge B \subseteq C^{\perp} \\
  &\wedge (\dim(B) = n/2 \vee \dim(C) = n/2) \\
  &\wedge (x \in L \vee \dim(B) = n/4 \vee \dim(C) = n/4)  \}
\end{align}

The scheme that we consider is the following:
\begin{enumerate}
\item The Prover and the Verifier have access to a CRS.
\item The Prover chooses a subspace $A$, computes $O_1 = \mathcal{O}(A)$, $O_2 = \mathcal{O}(A^\perp)$
\item The Prover computes a NIZK proof $\pi$ that $(x,O_1,O_2) \in L_{sub}$
\item The Prover sends $(O_1, O_2, \pi, \ket{A})$ to the Verifier
\item The Verifier accepts iff $\pi$ passes the ZK verification and $\ket{A}$ passes the verification with $O_1$ and $O_2$ (similarly to \Cref{fig:honest-verification}).
\end{enumerate}
\end{minipage}}
  \caption{Candidate anti-piracy proof system for $\NP$.}
  \label{fig:candidate-proof-system}
\end{figure}

We provide now the intuition on why we think that such a scheme is an anti-piracy proof system. First, notice that the pirate cannot trivially provide two copies of the proof because the state $\ket{A}$ is uncloneable. Moreover, the zero-knowledge proof alone cannot be used to verify that $x$ belongs to the language since soundness crucially requires checking the subspace state. Therefore, it seems that the quantum subspace state is making the cloneable part of the proof useless on its own.

We explain now the difficulties that we have in proving the security of this scheme. First of all, the only assumption that we have from admissible verifiers is completeness and soundness. Therefore, we do not have many handles to argue that one of the verifiers has no useful information on the subspace state.

Moreover, to prove the security of our protocol, we would need to prove that if the cryptographic objects that we use are secure, then either we can break the soundness of the protocol, or the language is easy (i.e. the pirate is cooking up the witness by herself). Proving the first conclusion is hard since the pirate can detect when we move from a yes-instance to a no-instance, or even harder, the pirate has a yes-instance in mind and creates an accepting good witness for a no-instance. 

To conclude the easiness of the language, we would need:
\begin{enumerate}
    \item the state provided to the pirate to be easily generated (instead of the original witness);
    \item even with the new state the pirate outputs a state that passes the admissible verification.
\end{enumerate}
The only route that we can see goes through using the simulation coming from the ZK property, which is done by creating a proof coupled with a simulated CRS. However, the state generated by the pirate would need to pass completeness and soundness using the real CRS.

\section{Complexity implications of cloneable proofs}   
\label{section:complexity_of_piracy}

In this section, we discuss the implications of the existence of cloneable proofs for $\QMA$, $\QMA(2)$, and, more generally, $\QMA(k)$.

\subsection{Cloneable quantum proofs}  

Recently, \cite{NZ24} defined the complexity class $\CQMA$ to denote the promise problems in $\QMA$ that have cloneable witnesses. 
We reproduce it below. 

\begin{definition}[$\CQMA_p$ \cite{NZ24}] 
A decision problem $L = (L_{\yes}, L_{\no})$ is said to be in the class $\CQMA(c,f,s)$ if there exists a polynomial-time quantum verifier $V$ (uniformly generated), a polynomial time quantum cloner $C$ (also uniformly generated), and a polynomial $p$, such that
\begin{itemize} 
\item \textbf{Completeness:} if $x \in L_{\yes}$, then there exists a quantum witness $\ket{\psi}$ on $p(|x|)$ qubits such that $\Pr[V(\ket{x},\ket{\psi}) = 1] \geq c$, and \\
\textbf{Cloning fidelity:} when given this same witness, $\ket{\psi}$, as input, $C$ succeeds at producing two independent copies of $\ket{\psi}$ with fidelity at least $f$. 
That is, 
\begin{equation} 
\bra{\psi} \otimes \bra{\psi} \; C(\kb{\psi} \otimes \kb{0}) \; \ket{\psi} \otimes \ket{\psi} \geq f. 
\end{equation} 
\item \textbf{Soundness:} if $x \in L_{\no}$, then for all quantum states $\ket{\psi^*}$ on $p(|x|)$ qubits, we have that $\Pr[V(\ket{x}, \ket{\psi^*}) = 1 ]  \leq s$. 
\end{itemize}
\end{definition}   
 
Since classical proofs are cloneable, we have $\QCMA$ (roughly, $\QMA$ with classical witnesses) is contained in $\CQMA$. 
Moreover, $\CQMA$ is contained in $\QMA$, by definition.
(Technically, we have $\CQMA(c,f,s) \subseteq \CQMA(c,f',s)$ where $f \geq f'$, and the equality $\QMA(c,s) = \CQMA(c,0,s)$.)  
Overall, we have the containments 
\begin{equation} 
\QCMA(c,s) \subseteq \CQMA(c,1,s) \subseteq \CQMA(c,0,s) = \QMA(c,s).  
\end{equation}   

\subsection{Cloneable unentangled quantum proofs} 

We now consider the similar setting where the proofs are unentangled, defined below. 
    
\begin{definition}[$\CQMAkp{p}$] 
A decision problem $L = (L_{\yes}, L_{\no})$ is said to be in the class $\CQMAkp{p}(c,f,s)$ if there exists a polynomial-time quantum verifier $V$ (uniformly generated), $k$ polynomial-time quantum cloners $C_1$, \ldots, $C_k$ (each uniformly generated), 
such that: 
\begin{itemize} 
\item \textbf{Completeness:} if $x \in L_{\yes}$, then there exists $k$ quantum witnesses $\ket{\psi_1}$, \ldots, $\ket{\psi_k}$, each on $p(|x|)$ qubits, such that
$\Pr[V(\ket{x},\ket{\psi_1},...,\ket{\psi_k}) = 1] \geq c$, and \\
\textbf{Cloning fidelity:} when given these same witnesses, $\ket{\psi_1}$, \ldots, $\ket{\psi_k}$, as input, $C_i$ succeeds at producing two independent copies of $\ket{\psi_i}$ with fidelity at least $f$. 
That is, 
\begin{equation} 
\bra{\psi_i} \otimes \bra{\psi_i} \; C_i(\kb{\psi_i} \otimes \kb{0}) \; \ket{\psi_i} \otimes \ket{\psi_i} \geq f, \; \text{ for all } i. 
\end{equation}  
\item \textbf{Soundness:} if $x \in L_{\no}$, then for all quantum states $\ket{\psi_1^*}$, \ldots, $\ket{\psi_k^*}$, each on $p(|x|)$ qubits, $\Pr[V(\ket{x},\ket{\psi_1^*},...,\ket{\psi_k^*}) = 1] \leq s$.
\end{itemize} 
\end{definition} 

\begin{remark}\label{rem:cloneable}
We make a few observations about the definition. 
\begin{itemize}
    \item First, note that, just like in $\QMA(k)$, there is no soundness guarantee for entangled proofs in $\CQMAk$.
    \item Second, similar to $\CQMA$, we are only interested in cloning ``good'' proofs, thus the cloning fidelity is only defined in the completeness conditions. 
    \item One could consider alternative definitions of $\CQMAk$. 
For instance, having a separate parameter $f_i$ for each cloning fidelity. 
This flexibility could find other applications, although it is not needed in this work. 
\item We have defined the cloners to act on separate proofs as separate circuits and not as a joint circuit. 
An alternative version with a joint cloner could find application, but one big obstacle in defining the (single) cloner this way is that an imperfect cloner could create entanglement in the proofs. 
While we do consider perfect cloning in this work, it is important to note that in our proofs below, we do end up applying the cloning circuits to proofs even when $x \in L_{no}$. 
Since there is no promise how a cloning circuit behaves for such no instances, it could create entanglement and ruin any soundness promises. 
\end{itemize}
\end{remark}

We begin by noting some immediate containments. 
Firstly, we have 
\begin{equation} 
\CQMAk(c,f,s) \subseteq \CQMAk(c,0,s) = \QMA(k). 
\end{equation} 
Secondly, we have 
\begin{equation} \label{eq24}
\CQMA_p(c,f,s) \subseteq \CQMAkp{p}(c,f,s)
\end{equation} 
since if $\ket{\psi}$ is a cloneable proof for $\QMA$, then $\ket{\psi}\ket{0}^{\otimes (k-1)}$ is a cloneable proof for $\QMA(k)$ using the same $\QMA$ verification and ignoring the $\ket{0}$ proofs.  
To summarize the complexity landscape so far, see the left figure in Figure~\ref{fig:complexity1} in \Cref{sec:intro-cloneable}.

In this section, we investigate the power of the class $\CQMAk$ with completeness-soundness gap $c-s \geq 1/q$, where $q$ is a polynomial-bounded function. 
We only consider cloning fidelity $f = 1$, i.e., perfect cloning. 
For brevity, we define 
\begin{align} 
\CQMAk & := \CQMAk(2/3,1,1/3) \\ 
\CQMA & := \CQMA(2/3,1,1/3). 
\end{align} 

The main result of this section is given in the following theorem.

\begin{theorem} \label{thm:cloneable}
We have 
\begin{equation}  
\CQMA(k) = \CQMA. 
\end{equation}
\end{theorem}

Determining any non-trivial relationship between $\QMA$ and $\QMAtwo = \QMA(k)$ is a big open problem in complexity theory. 
The above result settles this question in the cloneable setting. 
More precisely, when the proofs are cloneable, there is no added power from unentanglement.

\begin{corollary}
If $\CQMAtwo = \QMAtwo$, that is, $\QMAtwo$ can always be made to have cloneable proofs, then $\QMAtwo = \QMA$. 
Moreover, we would also have $\QMA = \CQMA$.  
\end{corollary}

\subsection{Proof of $\CQMAk = \CQMA$} 

In this section we prove \Cref{thm:cloneable} that $\QMA(k)$ and $\QMA$ are equal in the cloneable setting. 
The main idea of our proof is that we can use the completeness-soundness amplification techniques of~\cite{harrow2013testing} while maintaining \emph{similarly-sized} proofs. 
More precisely, we can obtain \emph{strong amplification} for $\CQMAtwo$ which, by a result of Brand\~{a}o (see~\cite{aaronson2008power}), will put it in $\QMA$. 
By careful observation, this reduction puts it into $\CQMA$ by noticing that we are maintaining cloneable proofs throughout the entire series of reductions. 

\begin{lemma}\label{lem:cloneable-amplification1}      
If $c-s \geq 1/q$, for some polynomial bounded function $q$, then we have  
\begin{equation}
\CQMAkp{p}(c,s) \subseteq \CQMAkp{p} \left( 1-2^{-\ell},1-\frac{1}{2q} \right)  
\end{equation} 
for any polynomial-bounded function $\ell$. 
\end{lemma} 

This follows from standard gap-amplification techniques (e.g., in~\cite{kobayashi2003quantum}) with the note that the proofs are still cloneable. 

\begin{proof} 
Choose the function $\ell$ as described in the lemma. 
We now take the verification circuit $V$ from the $\CQMAkp{p}(c,s)$ protocol and the cloning circuits $C_i$ and produce a new verification circuit for $\CQMAkp{p} \left( 1-2^{-\ell},1-\frac{1}{2q} \right)$. 
These will use the same proofs in the yes case, thus will still be cloneable. 
The rest follows from standard gap amplification techniques which we sketch below. 

The new verification takes the form:
\begin{enumerate}
\item Take the $k$ proofs $\ket{\psi^*_1} \otimes \cdots \otimes \ket{\psi^*_k}$ and apply the circuit $C_1 \otimes \cdots \otimes C_k$ to them $N:= 2\ell q^2$ times. 
\item Apply the verification circuit $V$ in parallel $(N+1)$ times. 
\item Accept if and only if $\left(\dfrac{c+s}{2}\right) \cdot (N+1)$ of these $N+1$ verifications accept. 
\end{enumerate} 
Using Chernoff and Markov inequalities, we get the desired completeness and soundness bounds.   
\end{proof} 

We now apply the product-state test from~\cite{harrow2013testing} to reduce the number of provers and to obtain a verification with an accepting POVM which is separable. 
The latter condition is crucial to amplify the completeness-soundness gap in a following lemma. 

\begin{lemma}\label{lem:cloneable-two-provers}
If $c > s$, we have 
\begin{equation} 
\CQMAkp{p}(c,s) \subseteq \CQMAtwosepp{kp} \left( \frac{1+c}{2}, 1 - \frac{(1-s)^2}{100} \right). 
\end{equation} 
Here, the SEP superscript means the accepting POVM of the $\QMAtwo$ verifier is separable. 
\end{lemma} 

This follows immediately from the analogous result  
in~\cite{harrow2013testing}\footnote{In particular, Lemma 5 in their updated arXiv.org preprint, version 6.} noting that cloneability is not required in the proof but carries throughout. 
More precisely, in the new verification, we replace the $k$ proofs $\ket{\psi_i}$ with two proofs, each being $\ket{\psi_1} \otimes \cdots \otimes \ket{\psi_k}$. 
Since each one is cloneable, then so is $\ket{\psi_1} \otimes \cdots \otimes \ket{\psi_k}$ (the new cloning circuit is $C_1 \otimes \cdots \otimes C_k$). 

\begin{lemma}\label{lem:cloneable-amplification2} 
We have  
\begin{equation}
\CQMAtwosepp{p}(c,s) \subseteq \CQMAtwop{p} (c^\ell, s^\ell) 
\end{equation} 
for any polynomial-bounded function $\ell$. 
\end{lemma} 

Again, this follows from the analogous result  
in~\cite{harrow2013testing}\footnote{In particular, Lemma 7 in their updated arXiv.org preprint, version 6.}, but in this case the size of the proof need not grow thanks to the ability to clone. 
We also only need it for the case $k=2$. 
  
\begin{proof} 
Let $V$ be a $\CQMAtwosepp{p}(c,s)$ verification and let $C_A$ and $C_B$ be the cloning circuits. 
We now define a verification for $\CQMAtwop{p}$ which does the following. 
First, it applies the cloning circuits $\ell$ times, on each respective proof, then it verifies all $\ell+1$ pairs of proofs each using the verification $V$. 
Overall, this new verification accepts if and only if all $\ell+1$ $V$ verifications accept. 

For completeness, it is straightforward to see that if you send a cloneable proof $\ket{\psi_A}\ket{\psi_B}$, then you will have $\ell+1$ copies afterwards, each of which will be accepted with probability at least $c^{\ell+1}$. 
Note that since this is the same proof, it is still cloneable. 

For soundness, this is a bit trickier since one has to be concerned with the verifications producing entanglement (for which $V$ may not be sound). 
We can analyze these in sequence. 
If a verification $V$ accepts, then since the accepting POVM is separable, the post-measured state (upon $V$ accepting) will not produce any entanglement. 
Therefore, the next set of (unentangled) proofs which are verified by $V$ will be accepted with probability at most $s$. 
Repeating this, the new soundness will be at most $s^{\ell+1}$, as desired. 
\end{proof} 

We remark that this proof is similar to the perfect parallel repetition result in~\cite{harrow2013testing} except that we can use the cloners to produce the $\ell+1$ copies of the proofs instead of asking the prover to supply them. 
We also remark that having a separable accepting POVM is crucial in this proof; we are unaware how to prove this claim without this assumption. 
We do not require this assumption in the rest of the proofs, thus we are not concerned with maintaining it. 

We now show that if a $\BQP$ verifier is sound against product states, then we can bound the soundness against all (even entangled) states as well. 
The following two lemmas follow closely from a result of Brand\~{a}o, found in a footnote in~\cite{aaronson2008power}. 

\begin{lemma} \label{useful}
Suppose $M$ is a positive semidefinite matrix. 
If $\bra{e} \bra{f} M \ket{e} \ket{f} \leq \alpha$ for all unit vectors $\ket{e}, \ket{f} \in \mathbb{C}^n$, then $\lambda_{\max}(M) \leq \alpha n^2$. 
\end{lemma} 

\begin{proof} 
Since $\bra{e} \bra{f} M \ket{e} \ket{f} = \| \sqrt{M} \ket{e} \ket{f} \|_2^2$, we have 
\begin{equation} \label{hyp} 
\| \sqrt{M} \ket{e} \ket{f} \|_2 \; \leq \; \sqrt{\alpha}, \text{ for all unit vectors } \ket{e}, \, \ket{f}. 
\end{equation}
Let $\ket{\psi}$ be the normalized eigenvector of $M$ corresponding to its greatest eigenvalue. 
It helps to write $\ket{\psi} = \sum_{i=1}^{n^2} \sqrt{\lambda_i} \ket{e_i} \ket{f_i}$ in its Schmidt decomposition (so $\lambda_i \geq 0$ and $\sum_{i=1}^{n^2} \lambda_i = 1$).  
We can write 
\begin{equation} 
\lambda_{\max}(M) 
= \bra{\psi} M \ket{\psi} 
= \| \sqrt{M} \ket{\psi} \|_2^2 
= \left\| \sqrt{M} \left( \sum_{i=1}^{n^2} \sqrt{\lambda_i} \ket{e_i} \ket{f_i} \right)  \right\|_2^2 
\leq \left( \sum_{i=1}^{n^2} \sqrt{\lambda_i} \left\| \sqrt{M} \ket{e_i} \ket{f_i} \right\|_2 \right)^2  
\end{equation} 
using the triangle inequality. 
We can now apply Eq.~\ref{hyp} to get 
\begin{equation} 
\lambda_{\max}(M) 
\leq \left( \sum_{i=1}^{n^2} \sqrt{\lambda_i} \sqrt{\alpha} \right)^2 
= \alpha \left( \sum_{i=1}^{n^2} \sqrt{\lambda_i} \right)^2 
\leq \alpha n^2
\end{equation} 
where the last inequality holds by Cauchy-Schwarz. 
\end{proof} 

We use this lemma to prove the following. 

\begin{lemma} \label{lem:clonable-single-proof} 
We have  
\begin{equation}
\CQMAtwop{p}(c,s) \subseteq \CQMAp{2p} \left( c, 2^{2p} \cdot s \right). 
\end{equation} 
\end{lemma} 

Before we prove this lemma, we remark that we only use this lemma in the case when $c > 2^{2p} \cdot s$. 

\begin{proof} 
Suppose we take a $\CQMAtwop{p}(c,s)$ verification. 
Define the new verification where we remove the ``unentangled'' constraint. 
Then in a yes-instance, there exists a cloneable proof which is still accepted with probability $c$ in the new verification (and still cloneable). 
In a no-instance, by Lemma~\ref{useful}, by adding entanglement, soundness can increase to at most $s \cdot (2^p)^2 = s \cdot 2^{2p}$. 
Note that this verification takes a $2p$-qubit proof now since the original verification took \emph{two} proofs, each of size $p$ qubits. 
\end{proof} 

We now put these lemmas together to prove $\CQMAk = \CQMA$. 
  
\begin{proof}[Proof of \Cref{thm:cloneable}]
    Starting from $\CQMAk(2/3,1,1/3)$ and for any polynomials $p_1$ and~$p_2$, we have from \Cref{lem:cloneable-amplification1,lem:cloneable-two-provers,lem:cloneable-amplification2} that
  \begin{align*}
       \CQMAk_p(2/3,1,1/3) \subseteq \CQMA_{kp}(2)\left((1-2^{-p_1(|x|)-1})^{p_2(|x|)},1, (1-1/3600)^{p_2(|x|)}\right).
  \end{align*} 
  By picking $p_1$ to be any polynomial and $p_2 = \Omega(kp)$, we have from \Cref{lem:clonable-single-proof} that 
    \begin{align*}
       \CQMAk_p(2/3,1,1/3) \subseteq \CQMA_{2kp}\left(1-\negl(|x|),1, \negl(|x|)\right). 
  \end{align*} 
  The opposite containment (ignoring proof sizes) has already been mentioned (\Cref{eq24}). 
\end{proof}
  
\subsection{Remarks and future work regarding \Cref{thm:cloneable}}
\label{sec:remarks-cloneable}
We first mention that one can define a new class \class{PirateableQMA} where a cloner (now referred to as a pirate) takes as input a quantum proof $\ket{\psi}$ and outputs {$\ket{\psi'} \ket{\psi''}$ where $\ket{\psi'}$ and $\ket{\psi''}$ are proofs for \emph{possibly different} verifications (of the same problem instance).}
We suspect some of our results can be generalized to this setting as well and we leave its analysis as an interesting open problem. 
Secondly, it would be interesting to see how \Cref{thm:cloneable} behaves in the less-than-perfect cloneability setting. 
We believe that setting the cloneability parameter to $f = 1 - \negl(|x|)$ should not change the results. 
This would be beneficial since perfect cloneability (when $f = 1$) does depend on the gate set, whereas $f = 1 - \negl(|x|)$ does not. 
One issue that arises without having perfect clones is that the output of the cloning circuit need not be product, and thus the use of Chernoff's bound in the proof of \Cref{lem:cloneable-amplification1} poses an issue. 
Lastly, we do note that there are implicit lower bounds on the cloning parameter $f$ however, assuming unproven conjectures. 
For example, we have $\CQMAk(c,0,s) = \QMA(k)(c,s)$, so if $\CQMAk(c,0,s) = \CQMA(c',0,s')$, for reasonable parameters $c'$ and $s'$,  then this would imply $\QMA(k) = \QMA$. 
This statement can be strengthened to $f = 1/\exp$ since the cloning circuit can just output a maximally mixed state, and thus always exists. 
Thus, we suspect there is a dichotomy-type theorem classifying the consequences of having $f$ belonging to certain intervals.  

Finally, we think it would be interesting to understand the interplay of the parameters $(c,f,s)$ in $\CQMA$-type protocols. In particular, how does completeness-soundness gap amplification protocols affect the cloning fidelity $f$? And can one amplify $f$ without sacrificing much of the completeness and soundness guarantees? 
These and other related questions which delve deeper into the power and flexibility of cloneable proof systems are very interesting and we leave them to future work.

\bibliographystyle{Files/bibtex/quasar.bst}
\bibliography{
Files/bibtex/quasar-full.bib,
Files/bibtex/quasar-more.bib, 
Files/bibtex/quasar.bib} 

\appendix

\section{Preliminaries}
\label{Section:prilim}

In this section, we review relevant definitions  from existing work. 

\subsection{Quantum Merlin Arthur ($\QMA$) and Related Classes}

In the field of computational complexity theory, the class $\QMA$ (Quantum Merlin Arthur) \cite{KSV02} is the quantum analog of the non-probabilistic complexity class $\NP$ (non-deterministic polynomial) or, more precisely, of the probabilistic complexity class $\MA$ (Merlin-Arthur).
  
\begin{definition}[$\QMA$]
$\QMA$ is the class of decision problems such that any language $L \subseteq \{0,1\}^*$ is in~$\QMA$ if and only if there is a polynomial-time bounded quantum verifier $V$ and a polynomial~$p(\cdot)$ such that:

\begin{itemize}
\item (Completeness) For every $x \in L$, there exists a $p(n)$-qubit quantum state $\ket{\psi}$ (called the proof or witness), where $|x|=n$, such that $\Pr[V(x,\ket{\psi}) = 1] \geq 2/3$.
\item (Soundness) For every $x \notin L$, and for any $p(n)$-qubit quantum state $\ket{\psi}$, we have that $\Pr[V(x,\ket{\psi}) = 1] \leq 1/3$.

\end{itemize}
\end{definition}

A variant of $\QMA$ is where the witness is restricted to classical strings, resulting in the class $\QCMA$. 

\begin{definition}[$\QCMA$]
$\QCMA$ is the class of decision problems such that any language $L \subseteq \{0,1\}^*$ is in~$\QCMA$ if and only if there is a polynomial-time bounded \emph{quantum} verifier $V$ and a polynomial~$p(\cdot)$ such that:

\begin{itemize}
\item (Completeness) For every $x \in L$, there exists a $p(n)$-bit classical string $y$ (called the proof or witness), where $|x|=n$, such that $\Pr[V(x,y) = 1] \geq 2/3$.
\item (Soundness) For every $x \notin L$, and for any $p(n)$-length classical string $y$, $\Pr[V(x,y) = 1] \leq 1/3$.

\end{itemize}
\end{definition}

We can also define the two-prover version of the class $\QMA$, called $\QMA(2)$ \cite{KM03}.

\begin{definition}[$\QMA(2)$]
$\QMA(2)$ is the class of decision problems such that any language $L \subseteq \{0,1\}^*$ is in~$\QMAtwo$ if and only if there is a polynomial-time bounded quantum verifier $V$ and a polynomial~$p(\cdot)$ such that:

\begin{itemize}
\item (Completeness) For every $x \in L$, there exists two $p(n)$-qubit quantum states $\ket{\psi_1}$ and $\ket{\psi_2}$, where $|x|=n$, such that $\Pr[V(\ket{x},\ket{\psi_1},\ket{\psi_2}) = 1] \geq 2/3$.
\item (Soundness) For every $x \notin L$, and for any two $p(n)$-qubit quantum states $\ket{\psi_1}$ and $\ket{\psi_2}$, where $|x|=n$,  $\Pr[V(\ket{x},\ket{\psi_1},\ket{\psi_2}) = 1] \leq 1/3$.
\end{itemize}
\end{definition}

We can define the complexity class $\QMA(k)$ in an analogous way, where the verification algorithm receives $k$ proofs $\ket{\psi_1}$, \ldots, $\ket{\psi_k}$. 
 
\end{document}